\documentclass[a4paper, 11pt]{article}

\usepackage{amsmath,amssymb,amsthm}

\usepackage[a4paper, margin=1.2in]{geometry}

\usepackage[lined,boxed]{algorithm2e}
\usepackage{graphicx}

\newtheorem{theorem}{Theorem}%[section]
\newtheorem{proposition}[theorem]{Proposition}
\newtheorem{lemma}[theorem]{Lemma}

\newtheorem{definition}[theorem]{Definition}

 \newcommand{\us}{\ensuremath{{\overline{s}}}}
 \newcommand{\ls}{\ensuremath{{\underline{s}}}}

\newcommand{\poly}{\ensuremath{\textnormal{poly}}}

\newcommand{\occ}{\ensuremath{{\rm{occ}}}}

 \newcommand{\ignore}[1]{}

\title{Unsatisfiable CNF Formulas need many Conflicts
  \thanks{Research is  supported by the SNF Grant 200021-118001/1}
}
\author{Dominik Scheder\\ 
  ETH Z\"urich\\
  \texttt{dscheder@inf.ethz.ch} \and
  Philipp Zumstein\\
  ETH Z\"urich\\
  \texttt{zuphilip@inf.ethz.ch}}
\begin{document}

\maketitle

\begin{abstract} A pair of clauses in a CNF formula constitutes a
  conflict if there is a variable that occurs positively in one clause
  and negatively in the other. A CNF formula without any conflicts is
  satisfiable. The Lov\'asz Local Lemma implies that a $k$-CNF formula
  %with only clauses of size exactly $k$ (a {\em $k$-CNF formula}), 
  is satisfiable if each clause conflicts with at most
  $\frac{2^k}{e}-1$ clauses. It does not, however, give any good bound
  on how many conflicts an unsatisfiable formula has globally. We show
  here that every unsatisfiable $k$-CNF formula requires
  $\Omega(2.69^k)$ conflicts and there exist unsatisfiable $k$-CNF
  formulas with $O(3.51^k)$ conflicts.
\end{abstract}

\section{Introduction}

A boolean formula in conjunctive normal form, a {\em CNF formula} for
short, is a conjunction (AND) of {\em clauses}, which are disjunctions
(OR) of literals. A {\em literal} is either a boolean variable $x$ or
its negation $\bar{x}$.  We assume that a clause does neither contain
the same literal twice nor a variable and its negation. A CNF formula
where each clause contains exactly $k$ literals is called a $k$-CNF
formula.  Satisfiability, the problem of deciding whether a CNF
formula is satisfiable, plays a major role in computer science.  How
can a $k$-CNF formula be unsatisfiable?  If $k$ is large, each clause
is extremely easy to satisfy individually.  However, it can be that
there are conflicts between the clauses, making it impossible to
satisfy all of them simultaneously.  If a $k$-CNF formula is
unsatisfiable,
then we expect that there are many conflicts.\\

To give a formal setup, we say two clauses {\em conflict} if there is
at least one variable that appears positively in one clause and
negatively in the other. For example, the two clauses $(x \vee y)$ and
$(\bar{x} \vee u)$ conflict, as well as $(x \vee y)$ and $(\bar{x}
\vee \bar{y})$ do.  Any CNF formula without the empty clause and
without any conflicts is satisfiable.  For a formula $F$ we define the
{\em conflict graph} $CG(F)$, whose vertices are the clauses of $F$,
and two clauses are connected by an edge if they conflict. $\Delta(F)$
denotes the maximum degree of $CG(F)$, and $e(F)$ the number of
conflicts in $F$, i.e., the number of edges in $CG(F)$.  In fact, any
$k$-CNF formula is satisfiable unless $\Delta(F)$ and $e(F)$ are
large. A quantitative result follows from the lopsided Lov\'asz Local
Lemma~\cite{ES1991,AS2000,LS2007}: A $k$-CNF formula $F$ is
satisfiable unless some clause conflicts with $\frac{2^k}{e}$ or more
clauses, i.e., unless $\Delta(F) \geq \frac{2^k}{e}$. Up to a constant
factor, this is tight: Consider the formula containing all $2^k$
clauses over the variables $x_1, \dots, x_k$, the {\em complete
  $k$-CNF formula} which we denote by $\mathcal{K}_k$.  It is
unsatisfiable, and
$\Delta(\mathcal{K}_k)=2^k-1$.\\

As its name suggests, the lopsided Lov\'asz Local Lemma implies a {\em
local} result.
%  : A $k$-CNF formula $F$ is satisfiable unless {\em
%  somewhere} in $F$ there are many conflicts. 
Our goal is to obtain a {\em global} result: $F$ is satisfiable unless the total number of
conflicts is {\em very} large. We define two functions
\begin{eqnarray*}
  lc(k) & := & \max \{d \in \mathbb{N}_0 \ | \ \textnormal{ every }
  k\textnormal{-CNF formula } F \textnormal{ with }
  \Delta(F) \leq d \textnormal{ is satisfiable} \} \ ,\\
  gc(k) & := & \max \{e \in \mathbb{N}_0 \ | \ \textnormal{ every }
  k\textnormal{-CNF formula } F \textnormal{ with }
  e(F) \leq e \textnormal{ is satisfiable} \} \ .
\end{eqnarray*}
The abbreviations $lc$ and $gc$ stand for {\em local conflicts} and
{\em global conflicts}, respectively. From the above discussion,
$\frac{2^k}{e}-1 \leq lc(k) \leq 2^k-2$, hence we know $lc(k)$ up to a
constant factor. In contrast, it does not seem to be easy to prove
nontrivial upper and lower bounds on $gc(k)$.  Certainly, $gc(k) \geq
lc(k) \geq \frac{2^k}{e}-1$ and $gc(k) \leq e(\mathcal{K}_k)-1 =
\binom{2^k}{2}-1$.  Ignoring constant factors, $gc(k)$ lies somewhere
between $2^k$ and $4^k$.  This leaves much space for improvement.
In~\cite{SZ2008}, we proved that $gc(k) \in \Omega(2.27^k)$ and $gc(k)
\leq \frac{4^k}{\log^3 k}{k}$. In this paper, we improve upon these
bounds.  Surprisingly, $gc(k)$ is \emph{exponentially} smaller than
$4^k$.

\begin{theorem}
  Any unsatisfiable $k$-CNF formula contains
  $\Omega\left(2.69^k\right)$ conflicts. There are unsatisfiable
  $k$-CNF formulas with $O\left(3.51^k\right)$ conflicts.
\label{main}
\end{theorem}

We obtain the lower bound by a more sophisticated application of the
idea we used in~\cite{SZ2008}. The upper bound follows from a
construction that is partially probabilistic, and inspired in parts by
Erd\H{o}s' construction in~\cite{Erdos1964} of sparse $k$-uniform
hypergraphs that are not $2$-colorable. To simplify notation, we view
formulas as sets of clauses, and clauses as sets of literals. Hence,
$|F|$ denotes the number of clauses in $F$. Still, we will sometimes
find it convenient to use the more traditional logic notation.

\subsection*{Related Work}

Let $F$ be a CNF formula and $u$ be a literal. We define $\occ_F(u) :=
| \{C \in F \ | \ u \in C\}|$. For a variable $x$, we write $d_F(x) =
\occ_F(x) + \occ_F(\bar{x})$ and call it the {\em degree} of $x$.
%counts the number of clauses containing the variable $x$, irrespective
%of its polarity.
We write $d(F) = \max_x d_F(x)$. It is easy to see
that for a $k$-CNF formula, $\Delta(F) \leq k(d(F)-1)$. Define
$$
f(k)  :=  \max \{d \in \mathbb{N}_0 \ | \ \textnormal{ every }
k\textnormal{-CNF formula } F \textnormal{ with }
d(F) \leq d \textnormal{ is satisfiable} \} \ .\\
$$ 
%The function $f(k)$ has been subject of some research.
By an
application of Hall's Theorem, Tovey~\cite{Tovey1984} showed that every
$k$-CNF formula $F$ with $d(F) \leq k$ is satisfiable, hence $f(k)
\geq k$. Later, Kratochv\'il, Savick\'y and Tuza~\cite{KST1993} showed
that $f(k) \geq \frac{2^k}{ek}$ and
%: In our terminology, they showed that
%$lc(k) \geq \frac{2^k}{e}-1$ and then used the fact that $\Delta(F)
%\leq k(d(F)-1)$.
%As for an upper bound, in~\cite{KST1993} the authors
%show that
$f(k) \leq 2^{k-1} - 2^{k-4} - 1$. The upper bound was improved by
Savick\'{y} and Sgall~\cite{SS2000} to $f(k) \in O(k^{-0.26}2^k)$, by
Hoory and Szeider~\cite{HS2006} to $f(k) \in
O\left(\frac{\log(k)2^k}{k}\right)$, and recently by
Gebauer~\cite{Gebauer2009} to $f(k) \leq \frac{2^{k+2}}{k}-1$,
closing the gap between lower and upper bound on $f(k)$ up to a
constant factor. Actually, we used the formulas constructed in~\cite{HS2006}
to prove the upper bound $gc(k) \leq
\frac{4^k}{\log^3 k}{k}$ in~\cite{SZ2008}.\\

\section{A First Attempt} \label{section-first-attempt}

We sketch a first attempt on proving a nontrivial lower bound on
$gc(k)$. Though this attempt does not succeed, it leads us to other
interesting questions, results, and finally proof methods which can be
used to prove a lower bound on $gc(k)$.  Let $F$ be a $k$-CNF formula and
$x$ a variable. Every clause containing $x$ conflicts with every
clause containing $\bar{x}$, thus $e(F) \geq \occ_F(x) \cdot
\occ_F(\bar{x})$. Furthermore,
\begin{eqnarray}
  e(F) \geq \frac{1}{k} \sum_x \occ_F(x) \cdot \occ_F(\bar{x})\,,
\end{eqnarray}
where the $\frac{1}{k}$ comes from the fact that each conflict might
be counted up to $k$ times, if two clauses contain several
complementary literals.  Every unsatisfiable $k$-CNF formula $F$
contains a variable $x$ with $d_F(x) \geq \frac{2^k}{ek}$.  If this
variable is {\em balanced}, i.e., $\occ_F(x)$ and $\occ_F(\bar{x})$
differ only in a polynomial factor in $k$, then $e(F) \geq
\frac{4^k}{\poly(k)}$. Indeed, in the formulas constructed
in~\cite{Gebauer2009}, all variables are balanced. The same holds for
the complete $k$-CNF formula $\mathcal{K}_k$.  It follows that when
trying to obtain an upper bound on $gc(k)$ that is exponentially
smaller than $4^k$, we should construct a very {\em unbalanced}
formula.  We ask the following question:

\begin{quotation}
  {\em Question:} Is there a number $a>1$ such that for every
  unsatisfiable $k$-CNF formula $F$ there is a variable with
  $\occ_F(x) \geq a^k$ and $\occ_F(\bar{x}) \geq a^k$?
\end{quotation}

%If this were the case, then $gc(k) \geq a^{2k}$.  Clearly, if we
%replace the ``and'' in the question by an ``or'', the answer is yes,
%for any $a < 2$. With the ``and'', however, 
The answer is a very
strong no:
%We can construct an unsatisfiable 
%$k$-CNF formula where $\occ_F(\bar{x})$ is arbitrarily small for {\em
%  every} $x$, but at the price of making $\occ_F(x)$ very large for
%{\em at least one} $x$.
In~\cite{SZ2008} we gave a simple inductive construction of a $k$-CNF
formula $F$ with $\occ_F(\bar{x})\leq 1$ for every variable $x$.
However, in this formula one has $\occ_F(x) \approx k!$.
%\footnote{
%  If this is not the
%  case, we can simply replace every occurrence of $x$ by $\bar{x}$ and
%  vice versa, without changing the structure of the formula.}
Allowing $\occ_F(\bar{x})$ to be a small exponential in $k$, we have
the following result:

\begin{theorem}
 
  \begin{itemize}
  \item[(i)] For every $a > 1$, $b \geq \frac{a}{a-1}$ there is a
    constant $c$ such that for all sufficiently large $k$, there is an
    unsatisfiable $k$-CNF formula $F$ with $\occ_F(\bar{x}) \leq
    ck^2a^k$ and $\occ_F(x) \leq ck^2b^k$, for all $x$.
  \item[(ii)] Let $1 < a < \sqrt{2}$ and $b =
    \sqrt{\frac{a^4}{a^2-1}}$. Then every $k$-CNF formula $F$ with
    $\occ_F(x) \leq \frac{b^k}{8k}$ and 
    $\occ_F(\bar{x}) \leq \frac{a^k}{8k}$ is
    satisfiable.
  \end{itemize}
  \label{tradeoff-exponential}
\end{theorem}

Of course, we can interchange the
roles of $x$ and $\bar{x}$, but it is convenient to assume that
$\occ_F(\bar{x}) \leq \occ_F(x)$ for every $x$.
In the spirit of these
results, we might suspect that if $F$ is unsatisfiable, then for some
variable $x$, the product  $\occ_F(x) \cdot \occ_F(\bar{x})$ is large.

\begin{quotation}
  {\em Question:} Is there a number $a > 2$ such that every unsatisfiable
  $k$-CNF formula contains a variable $x$ with 
  $\occ_F(x) \cdot \occ_F(\bar{x}) \geq a^k$?
\end{quotation}

Clearly, $gc(k) \geq a^k$ for any such number $a$. The complete
$k$-CNF formula witnesses that $a$ cannot be greater than $4$, and it
is not at all easy to come up with an unsatisfiable $k$-CNF formula
where $\occ_F(x)\cdot\occ_F(\bar{x})$ is exponentially smaller than
$4^k$ for every $x$. We cannot answer the above question, but we
suspect that the answer is yes. We prove an upper bound on the
possible value of $a$:

\begin{theorem}
  There are unsatisfiable $k$-CNF formulas with $\occ_F(x) \cdot
  \occ_F(\bar{x}) \in O(3.01^k)$ for all variables $x$.
\label{theorem-individual}
\end{theorem}

%%%%%%%%%%%%%%%%%%%%%%%%%%%%%%%%%%%%%%%%%%%%%%%%%%%%%%%%%%%%%
\section{Proofs}\label{section-proofs}
%%%%%%%%%%%%%%%%%%%%%%%%%%%%%%%%%%%%%%%%%%%%%%%%%%%%%%%%%%%%%

For a truth assignment $\alpha$ and a clause $C$, we will write
$\alpha \models C$ if $\alpha$ satisfies $C$, and $\alpha\not\models
C$ if it does not.  Similarly, if $\alpha$ satisfies a formula $F$, we
write $\alpha \models F$. We begin by stating a version of the
Lopsided Lov\'asz Local Lemma formulated in terms of satisfiability.
See~\cite{SZ2008} for a derivation of this version.

\begin{lemma}[SAT version of the Lopsided Lov\'asz Local Lemma]
  Let $F$ be a CNF formula not containing the empty clause.  Sample a
  truth assignment $\alpha$ by independently setting each variable $x$
  to {\em \texttt{true}} with some probability $p(x) \in [0,1]$.  If
  for any clause $C \in F$, it holds that
 \begin{eqnarray}
   \sum_{D \in F:\ C\textnormal{ and } D \textnormal{ conflict}} \Pr[\alpha 
   \not \models D] \leq \frac{1}{4}\,,
   \label{ineq-sum}
 \end{eqnarray}
 then $F$ is satisfiable.
\label{corollary-llll}
\end{lemma}

It is not possible to apply
Lemma~\ref{corollary-llll} directly to a formula $F$ which we want to prove
being satisfiable. Instead, we apply it to a formula $F'$ we obtain
from $F$ in the following way:

\begin{definition}
  Let $F$ be a CNF formula. A {\em truncation} of $F$ is a CNF formula
  $F'$ that is obtained from $F$ by deleting some literals from some
  clauses.
\end{definition}

For example, $(x \vee y) \wedge (\bar{y} \vee z)$ is a truncation of
$(x \vee y \vee \bar{z}) \wedge (\bar{x} \vee \bar{y} \vee z)$.  A
truncation of a $k$-CNF formula is not necessarily a $k$-CNF formula
anymore. Any truth assignment satisfying a truncation $F'$ of $F$ also
satisfies $F$. In our proofs, we will often find it easier to apply
Lemma~\ref{corollary-llll} to a special truncation of $F$ than to $F$
itself. We need a technical lemma on the binomial coefficient.
 
\begin{lemma}\label{lem:technical}
  Let $a,b\in \mathbb{N}$ with $b/a\leq 0.75$. Then
  \[
  \frac{a^b}{b!} \geq \binom{a}{b} > \frac{a^b}{b!} e^{-b^2/a}\,.
  \]
\end{lemma}

\begin{proof}
  The upper bound is trivial and true for all $a,b$. The lower bound
  follows like this.
\[
\binom{a}{b} = \frac{a(a-1)\cdots 1}{b!} = \frac{a^b}{b!}
\prod_{j=0}^{b-1} \frac{a-j}{a} > \frac{a^b}{b!} e^{-2/a
  \sum_{j=0}^{b-1} j} \geq \frac{a^b}{b!} e^{-b^2/a}\,,
\]
where we used the fact that $1-x > e^{-2x}$ for $0\leq x \leq 0.75$.
\end{proof}

%\subsection{Proof of Theorem}
%
%
%
%

\subsection{Proof of Theorem~\ref{tradeoff-exponential} and~\ref{theorem-individual}}

As we have argued in Section~\ref{section-first-attempt}, in order to
improve significantly upon the upper bound $gc(k) \leq 4^k$, we must
construct a formula that is very unbalanced, i.e.  $\occ_F(x)$ is
exponentially larger than $\occ_F(\bar{x})$.
First, we will construct an unsatisfiable CNF formula with $k$-clauses
and some smaller clauses.
%The central idea is that
%we do not construct an unsatisfiable $k$-CNF formula, but allow
%certain clauses to be smaller.
In a second step, we expand all
clauses to size $k$.

% Next, we will show a simple construction of an unsatisfiable $k$-CNF
% which will upper bound the number of conflicts a single variables causes.
% The central idea is that somehow we must construct a very unbalanced
% formula, i.e.  one where $x$ occurs much more often than $\bar{x}$,
% for at least one variable.  We do this as follows: Fix some $\alpha
% \in [0,1]$, to be determined later. If $\alpha k$ is not an integer,
% we round down. We choose a set $V = \{x_1, \dots,
% x_{(1+\alpha)k}\}$ of $(1+\alpha)k$ variables. Let $\bar{V}$ denote
% the set of all negative literals. We define
% \begin{equation}\label{eqn:Fak}
% 	F_{\alpha,k} := \binom{V}{\alpha k} \cup \binom{\bar{V}}{k} .
% \end{equation}
% $F_{\alpha,k}$ contains all positive $\alpha k$-clauses and all negative
% $k$-clauses over $V$.
% The formula is unsatisfiable: A truth assignment either sets at least
% $\alpha k$ variables to \texttt{false}, and some clause in $\binom{V}{
% \alpha k}$ is unsatisfied, or it sets at least $k$ variables to
% $\texttt{true}$, and some clause in $\binom{\bar{V}}{k}$ is
% unsatisfied. Further, $F_{\alpha,k}$ has $\binom{(1+\alpha)k}{\alpha k}$
% positive and $\binom{(1+\alpha)k}{\alpha k}$ negative clauses. It is
% not yet a $k$-CNF formula. 

\begin{definition}
  Let $F$ be a CNF formula with clauses of size at most $k$.  For each
  $k'$-clause $C$ with $k'<k$, construct a complete $(k-k')$-CNF
  formula $\mathcal{K}_{k-k'}$ over $k-k'$ new variables
  $y^C_1,\dots,y^C_{k-k'}$. We replace $C$ by $C \vee
  \mathcal{K}_{k-k'}$. Using distributivity, we expand it into a
  $k$-CNF formula $G$ called a \emph{$k$-CNFification of $F$}.
\end{definition}

For example, a $3$-CNFification of $(x \vee y) \wedge (\bar{x} \vee
y \vee z)$ is $(x \vee y \vee y_1) \wedge (x \vee y \vee \bar{y}_1)
\wedge (\bar{x} \vee y \vee z)$. A truth
assignment satisfies $F$ if and only if it satisfies its
$k$-CNFification $G$.

\begin{definition}
  Let $\ell, k \in \mathbb{N}_0$. An $(\ell,k)$-CNF formula is a
  formula consisting of $\ell$-clauses containing only positive
  literals, and $k$-clauses containing only negative literals.
\end{definition}

If $F$ is an $(\ell,k)$-CNF formula, we write $F=F^+ \wedge F^-$,
where $F^+$ consists of the positive $\ell$-clauses and $F^-$ of
the negative $k$-clauses.

\begin{proposition}
  Let $\ell \leq k$ and let $F=F^+ \wedge F^-$ be an $(\ell,k)$-CNF
  formula.  Let $G$ be the $k$-CNFification of $F$. Then
  \begin{itemize}
    \item[(i)] $e(G) \leq 4^{k-\ell}|F^+| + 2^{k-\ell}|F^+|\cdot|F^-|$,
    \item[(ii)] $\occ_G(x) \cdot \occ_G(\bar{x}) \leq
      \max\{4^{k-\ell}, 2^{k-\ell}|F^+|\cdot|F^-|\}$.
  \end{itemize}
\label{prop-lk}
\end{proposition}

\begin{proof}
  To prove (i), note that every edge in $CG(F)$ runs between a
  positive $\ell$-clause $C$ and a negative $k$-clause $D$. Thus,
  $e(F) \leq |F^+|\cdot|F^-|$.  In $G$, this edge is replaced by
  $2^{k-\ell}$ edges, since $C$ is replaced by $2^{k-\ell}$ copies.
  This explains the term $2^{k-\ell}|F^+|\cdot|F^-|$. Replacing $C$ by
  $2^{k-\ell}$ many $k$-clauses introduces at most $4^{k-\ell}$ new
  conflicts. This explains the term $4^{k-\ell}|F^+|$, and proves (i).
  To prove $(ii)$, there are two cases. First, if $x$ appears in $F$,
  then $\occ_G(\bar{x}) = \occ_F(\bar{x})$ and $\occ_G(x) =
  \occ_F(x)2^{k-\ell}$, thus $\occ_G(x)\occ_G(\bar{x}) \leq
  2^{k-\ell}|F^+|\cdot|F^-|$. Second, if $x$ appears in $G$, but not
  in $F$, then $\occ_G(x) = \occ_G(\bar{x}) = 2^{k-\ell-1}$, and
  $\occ_G(x) \cdot \occ_G(\bar{x}) \leq 4^{k-\ell}$.
\end{proof}

We should explore for which values of $|F^+|$ and $|F^-|$ there are
unsatisfiable $(\ell,k)$-CNF formulas. We can then use
Proposition~\ref{prop-lk} to derive upper bounds.

% Here is a
% simple idea: Take $k+\ell-1$ variables and let $F$ consist of all
% positive $\ell$-clauses and negative $k$-clauses over this set.
% Consider any assignment. Either it sets at least $\ell$ variables to
% \texttt{false}, thus unsatisfying a positive clause, or at least $k$
% variable to \texttt{true}, unsatisfying a negative clause. Hence $F$
% is unsatisfiable and contains $\binom{k+\ell-1}{k}+\binom{k+\ell-1}
%   {\ell}$ clauses. For $k=\ell$, this yields a $(k,k)$-CNF formula with
% $\binom{2k-1}{k}\in \Theta \left(\frac{4^k}{\sqrt{k}}\right)$ clauses.
% Is this optimal? 
% For $\ell=k$, this is esentially asking for a non-$2$-colorable
% $k$-uniform hypergraph with a small number of edges. This question was
% raised by Erd\H{o}s~\cite{Erdos1963} in 1963. Clearly, such a
% hypergraph contains at least $2^{k-1}$ edges. This was already
% observed in~\cite{Erdos1963}.  One year later, he~\cite{Erdos1964}
% gave a probabilistic construction of an non-$2$-colorable $k$-uniform
% hypergraph with only $ck^2 2^k$ edges. This bound has been improved by
% Beck~\cite{Beck1978} to $r^{\frac{1}{3}-o(1)}2^r$ and later by
% Radhakrishnan and Srinivasan~\cite{RS2000} to $0.7\sqrt{\frac{r}{\ln
%     r}}2^r$. A minute of thought shows that the size of a smallest
% non-$2$-colorable $k$-uniform hypergraph and that of a smallest
% unsatisfiable $(k,k)$-CNF formula is, up to a factor of two, the same.
% We adapt Erd\H{o}s' construction~\cite{Erdos1964} to $(\ell,k)$-CNF
% formulas.  Please think of $\ell$ as a constant fraction of $k$, say
% $\ell = \alpha k$ for some constant $\alpha\in [0,1]$.

\begin{lemma}\label{lem:t}
  For any $\rho \in (0,1)$, there is a constant $c$ such that for all
  $k\in\mathbb{N}_0$ and $\ell \leq k$, there exists an unsatisfiable
  $(\ell, k)$-CNF formula $F = F^+ \wedge F^-$ with $|F^-| \leq ck^2
  \rho^{-k}$ and $|F^+| \leq ck^2 (1-\rho)^{-\ell}$. \\
\end{lemma}

\begin{proof}
  We choose a set variables $V=\{x_1,\dots,x_n\}$ of $n=k^2$
  variables.  There are $\binom{n}{k}$ $k$-clauses over $V$ containing
  only negative literals. We form $F^-$ by sampling $ck^2 \rho^{-k}$
  of them, uniformly with replacement, and similarly, we form $F^+$ by
  sampling $ck^2 (1-\rho)^{-\ell}$ purely positive $\ell$-clauses,
  where $c$ is some suitable constant determined later. Set $F=F^-
  \wedge F^+$.  We claim that with high probability, $F$ is
  unsatisfiable.  Let $\alpha$ be any truth assignment. There are two
  cases.
  
  \emph{Case 1.} $\alpha$ sets at least $\rho n$ variables to
  \texttt{true}.  For a random negative clause $C$,
  \[
  \Pr[\alpha \not \models C] \geq 
  \frac{\binom{\rho n}{k}}{\binom{n}{k}} \geq c' \rho^k\,,
  \]
  The last inequality follows from Lemma~\ref{lem:technical}. Since we
  select the clauses of $F^-$ independently of each other, we obtain
  \[
  \Pr[\alpha \models F^-] \leq (1 - c'\rho^k)^{ck^2 \rho^{-k}} <
  e^{-c c' k^2}=e^{-k^2}\, ,
  \]
  provided we chose $c$ large enough, i.e., $c \geq \frac{1}{c'}$.
  
  {\em Case 2}: $\alpha$ sets at most $\frac{n}{a}$ variables to
  \texttt{true}. Now a similar calculation shows that $\alpha$
  satisfies $F^+$ with probability at most $e^{-k^2}$.

  In any case, $\Pr[\alpha \models F] \leq e^{-k^2}$. The expected
  number of satisfying assignments of $F$ is at most $2^{k^2}e^{-k^2}
  \ll 1$ and with high probability $F$ is unsatisfiable.
\end{proof}

The bound in Lemma~\ref{lem:t} is tight up to a polynomial factor in
$k$:

\begin{lemma}
 Let $F=F^+ \wedge F^-$ be an $(\ell,k)$-CNF formula.
  If there is a $\rho \in (0,1)$ such that $|F^+| <
  \frac{1}{2}(1-\rho)^{-\ell}$ and $|F^-| < \frac{1}{2}\rho^{-k}$,
  then $F$ is satisfiable.
\end{lemma}

\begin{proof}
  Sample a truth assignment
  $\alpha$ by setting each variable independently to \texttt{true}
  with probability $\rho$. For a negative $k$-clause $C$, it holds
  that $\Pr[\alpha \not \models C] = \rho^k$. Similarly, for a
  positive $\ell$-clause $D$, $\Pr[\alpha \not \models D] =
  (1-\rho)^{\ell}$.  Hence the expected number of clauses in $F$ that
  are unsatisfied by $\alpha$ is $ \rho^k |F^-| + (1-\rho)^{\ell}|F^+|
  < \frac{1}{2} + \frac{1}{2} = 1$.
  Therefore, with positive probability $\alpha$ satisfies $F$.
\end{proof}

%  {\em Case 1}: $\alpha$ sets at least $\frac{n}{a}$ variables to
%  \texttt{true}. Let $C$ be a negative $k$-clause, chosen uniformly at
%  random.
%  $$
%  \Pr[\alpha \not \models C] \geq \frac{\binom{n/a}{k}}{\binom{n}{k}}
%  \geq \frac{c'}{a^k} \ ,
%  $$
%  with the constant $c'=e^{-a}$. The last
%  inequality follows from Lemma~\ref{lem:technical}. Since we select
%  the clauses of $F^-$ independently of each other, we obtain
%  $$
%  \Pr[\alpha \models F^-] \leq \left(1 - \frac{c'}{a^k}\right)^{ck^2a^k}
%  \leq e^{-k^2} \ , 
%  $$
%  where $c=1/c'$.
%
%  {\em Case 2}: $\alpha$ sets at most $\frac{n}{a}$ variables to
%  \texttt{true}. Now a similar calculation shows that
%  $\alpha$ satisfies $F^+$ with probability at most $e^{-k^2}$.
%
%  In any case, $\Pr[\alpha \models F] \leq e^{-k^2}$. The expected
%  number of satisfying assignments of $F$ is at most $2^{k^2}e^{-k^2}
%  \ll 1$ and with high probability $F$ is unsatisfiable proving $(i)$.\\
%

%It should be noted that the quantitative bounds in
%Theorem~\ref{tradeoff-exponential} can be improved by
%slight modifications of the proof. Here, we are more interested in
%proving a qualitative tradeoff and therefore prefer to give a simple
%proof with a worse bound. In any case, our methods seem unlikely to
%bring the bounds in point $(i)$ and $(ii)$ of
%Theorem~\ref{tradeoff-exponential} close to each other.

\begin{proof}[Proof of Theorem~\ref{tradeoff-exponential}]
  $(i)$ Apply Lemma~\ref{lem:t} with $\ell=k$ and $\rho=\frac{1}{a}$.\\
  
  $(ii)$ We fix some probability $p := \frac{1}{a^2} \geq
  \frac{1}{2}$, and set every variable of $F$ to \texttt{true} with
  probability $p$, independent of each other. This gives a random
  truth assignment $\alpha$. We define a truncation $F'$ of $F$ as
  follows: For each clause $C \in F$, if at least half the literals of
  $C$ are negative, we remove all positive literals from $C$ and
  insert the truncated clause into $F'$, otherwise we insert $C$ into
  $F'$ without truncating it.  We write $F' = F_k \wedge F^-$, where
  $F^-$ consists of purely negative clauses of size at least
  $\frac{k}{2}$, and $F_k$ consists of $k$-clauses, each containing at
  least $\frac{k}{2}$ positive literals.  A clause in $F^-$ is
  unsatisfied with probability at most $p^{\frac{k}{2}}$, and a clause
  in $F_k$ with probability at most $p^{\frac{k}{2}}
  (1-p)^{\frac{k}{2}}$. This is because in the worst case, half of all
  literals are negative: Since $p \geq \frac{1}{2}$, negative literals
  are more likely to be unsatisfied than positive ones.  Let $C \in
  F'$ be any clause. A positive literal $x \in C$ causes conflicts
  between $C$ and the $\occ_{F'}(\bar{x}) \leq \frac{a^k}{8k}$ clauses
  of $F'$ containing $\bar{x}$. Similarly, a negative literal $\bar{y}
  \in C$ causes conflicts with the at most $\frac{b^k}{8k}$ clauses of
  $F_k$ containing $y$. Therefore
  $$
  \sum_{D\in F : \text{ $C$ and $D$ conflict}}\Pr[\alpha \not \models D] \leq
  \frac{a^k}{8} p^{\frac{k}{2}} + \frac{b^k}{8}
  p^{\frac{k}{2}}(1-p)^{\frac{k}{2}} = \frac{1}{4} \ ,
  $$
  since $p=\frac{1}{a^2}$ and 
  $b = \sqrt{\frac{a^4}{a^2-1}}$.
  By Lemma~\ref{corollary-llll}, $F'$ is satisfiable. 
\end{proof}

Part $(ii)$ of Theorem~\ref{tradeoff-exponential} can easily be improved
by defining a more careful truncation procedure: We remove all
positive literals from a clause $C$ if $C$ contains less than $\lambda
k$ of them, for some $\lambda \in [0,1]$. Choosing $\lambda$ and $p$
optimally, we obtain a better result, but the calculations become
messy, and it offers no additional insight. The crucial part of the
proof is that by removing positive literals from a clause, we can use
the fact that $\occ_F(\bar{x})$ is small to bound the number of
clauses $D$ that conflict with $C$ and have a large probability of
being unsatisfied.  This is also the main idea in our proof of the
lower bound of Theorem~\ref{main}.
It should be pointed out that for $k = \ell$, an $(\ell,k)$-CNF
formula is just a monotone $k$-CNF formula. The size of a smallest
unsatisfiable monotone $k$-CNF formula is the same---up to a factor of
at most $2$---as the minimum number of hyperedges in a $k$-uniform
hypergraph that is not $2$-colorable. In 1963,
Erd\H{o}s~\cite{Erdos1963} raised the question what this number is,
and proved lower bound of $2^{k-1}$ (this is easy, simple choose a
random $2$-coloring). One year later, he~\cite{Erdos1964} gave a
probabilistic construction of a non-$2$-colorable $k$-uniform
hypergraph using $ck^22^k$ hyperedges. For $\ell=k$ and
$\rho=\frac{1}{2}$, the statement and proof of Lemma~\ref{lem:t} are
basically the same in~\cite{Erdos1964}.

 \begin{proof}[Proof of Theorem~\ref{theorem-individual}]
   Combining Lemma~\ref{lem:t} and Proposition~\ref{prop-lk}, we
   conclude that for any $\rho \in (0,1)$ and $0 \leq \ell \leq k$,
   there is an unsatisfiable $k$-CNF formula $F$ with
   $$
   \occ_F(x) \cdot \occ_F(\bar{x}) \leq
   \max \{4^{k-\ell}, 2^{k-\ell} c^2k^4 \rho^{-k} (1-\rho)^{-\ell}\} \ ,
   $$
   for every variable $x$. The constant $c$ depends on $\rho$, but not
   on $k$ or $\ell$. For fixed $k,\ell>1$, the term $\rho^{-k}
   (1-\rho)^{-\ell}$ is minimized for $\rho = \frac{k}{k+\ell}$.
   Choosing $\ell = \left\lceil 0.2055k\right\rceil$, we get $\rho
   \approx 0.83$ and $\occ_F(x) \cdot \occ_F(\bar{x}) \in O(3.01^k)$.
 \end{proof}

%%%%%%%%%%%%%%%%%%%%%%%%%%%%%%%%%%%%%%%%%%%%%%%%%%%%%%%%%%%%%%%%
\subsection{Poof of the Main Theorem}\label{section-lower-bound}
%%%%%%%%%%%%%%%%%%%%%%%%%%%%%%%%%%%%%%%%%%%%%%%%%%%%%%%%%%%%%%%%

 \begin{proof}[Proof of the upper bound of Theorem~\ref{main}]
   As in the previous proof, Proposition~\ref{prop-lk} together
   with Lemma~\ref{lem:t} yield an unsatisfiable $k$-CNF formula $F$ 
   with 
   $$
   e(F) \leq 4^{k - \ell} ck^2(1-\rho)^{-\ell} + 
   2^{k-\ell} c^2k^4\rho^{-k}(1-\rho)^{-\ell}\ .
   $$
   For $\rho \approx 0.6298$ and $\ell = \left\lceil
     0.333k\right\rceil$, we obtain $e(F) \in O(3.51^k)$.
 \end{proof}

\begin{proof}[Proof of the lower bound in Theorem~\ref{main}]
  Let $F$ be an unsatisfiable $k$-CNF and let $e(F)$ be the number of
  conflicts in $F$. We will show that $e(F) \in
  \Omega\left(2.69^k\right)$. In the proof, $x$ denotes a variable and
  $u$ a positive or negative literal.  We assume $\occ_F(\bar{x}) \leq
  \occ_F(x)$ for all variables $x$. We can do so since otherwise we
  just replace $x$ by $\bar{x}$ and vice versa. This changes neither
  $e(F)$, nor satisfiability of $F$. Also we can assume that
  $\occ_F(x)$ and $\occ_F(\bar{x})$ are both at least $1$, if $x$
  occurs in $F$ at all. For $x$, we define
  \[
  p(x) := \max \left\{\frac{1}{2}, 
    \sqrt[k]{\frac{\occ_F(x)}{16e(F)}}\ \right\}\,,
  \]
  and set $x$ to \texttt{true} with probability $p(x)$ independently
  of all other variables yielding a random assignment $\alpha$. Since
  $\occ_F(u) \leq e(F)$, we have $p(x) \leq 1$. We set $p(\bar{x}) = 1
  - p(x)$. By definition, $p(x) \geq p(\bar{x})$.  Let us list some
  properties of this distribution.  First, if $p(u) < \frac{1}{2}$ for
  some literal $u$, then $u$ is a negative literal $\bar{x}$, and
  $p(x) = \sqrt[k]{\frac{\occ_F(x)}{16e(F)}} > \frac{1}{2}$. Second,
  if $p(u) = \frac{1}{2}$, then both
  $\sqrt[k]{\frac{\occ_F(u)}{16e(F)}} \leq \frac{1}{2}$ and
  $\sqrt[k]{\frac{\occ_F(\bar{u})}{16e(F)}} \leq \frac{1}{2}$ hold.
  We distinguish two types of clauses: {\em Bad} clauses, which
  contain at least one literal $u$ with $p(u) < \frac{1}{2}$, and {\em
    good} clauses, which contain only literals $u$ with $p(u) \geq
  \frac{1}{2}$.
  %A good clause is unsatisfied with probability at most
  %$2^{-k}$, but for bad clauses, the probability might be rather
  %large.
  Let $\mathcal{B} \subseteq F$ denote the set of bad clauses and $\mathcal{G}\subseteq F$ the set of good clauses.

  \begin{lemma}
   $\sum_{C \in \mathcal{B}} \Pr\ [\alpha \not \models C] \leq \frac{1}{8}$.
   \label{prop-bad}
  \end{lemma}
  \begin{proof}
    For each clause $C \in \mathcal{B}$, let $u_C$ be the literal in
    $C$ minimizing $p(u)$, breaking ties arbitrarily. This means
    $\Pr[\alpha \not \models C] \leq p(\bar{u}_C)^k$. Since $C$ is a
    bad clause, $p(u_C) < \frac{1}{2}$, $u_C$ is a negative literal
    $\bar{x}_C$, and $p(x_C) = \sqrt[k]{\frac{\occ_F(x_C)}{16e(F)}}$.
    Thus
    \begin{eqnarray}
      \sum_{C \in \mathcal{B}} \Pr[\alpha
      \not \models C] \leq \sum_{C \in \mathcal{B}}
      p(x_C)^k = \sum_{C \in \mathcal{B}} \frac{\occ_F(x_C)}{16e(F)} \ .
      \label{ineq-bad-2}
    \end{eqnarray}
    Since clause $C$ contains $\bar{x}_C$, it conflicts with all
    $\occ_F(x_C)$ clauses containing $x_C$, thus $\sum_{C \in
      \mathcal{B}} \occ_F(x_C) \leq 2e(F)$.  The factor $2$ arises
    since we count each conflict possibly twice, once from each side.
    Combining this with (\ref{ineq-bad-2}) proves the lemma.
  \end{proof}

%   Consider the following scenario: A clause $C$ contains a literal
%   $\bar{x}$ with $p(x)=0.9$, i.e.  $\occ_F(x) \gg 2^k$. Every clause
%   containing $x$ is good and, besides $x$, contains only literals that
%   are satisfied with probability $\frac{1}{2}$.  Then the sum in
%   (\ref{ineq-sum}) is at least $\occ_F(x) 2^{-k+1} \frac{1}{10} \gg
%   \frac{1}{4}$, hence we cannot apply Lemma~\ref{corollary-llll}.
%   To circumvent scenarios like this one, we first do a sparsification
%   step. Let $\mathcal{G} \subseteq F$ denote the set of all good
%   clauses. 

  We cannot directly apply Lemma~\ref{corollary-llll} to $F$.
  Therefore we apply the below sparsification process to $F$.
%  \domi{check exact spot of algorithm on the page}
\begin{algorithm}[h]
\Titleofalgo{\ Sparsification Process}
\SetTitleSty{}{}
Let $\mathcal{G'}=\{D \in F \ | \  p(u) \geq \frac 12, \forall u \in D\}$
be the set of good clauses in $F$.\\
%$\mathcal{G}' := \mathcal{G}$\\
\vspace{1mm}
 \While{$\exists \textnormal{ a literal } u:  \ 
   \sum_{D\in \mathcal{G}' : u \in D } 
   \Pr[\alpha \not \models D] > \frac{1}{8k}$}
 {
   \vspace{1mm}
   Let $C$ be some clause maximizing 
   $\Pr[\alpha \not \models D]$
   among all clauses $D \in \mathcal{G}': u \in D$.\\
   \vspace{1mm}
   $C' := C \setminus \{u\}$\\
   $\mathcal{G}' := (\mathcal{G}' \setminus \{C\}) \cup \{C'\}$
   \vspace{1mm}
 }
 \Return{$F' := \mathcal{G'} \cup \mathcal{B}$}
\end{algorithm}

\begin{lemma}
  Let $F'$ be the result of the sparsification process. If $F'$ does
  not contain the empty clause, then $F$ is satisfiable.
  \label{prop-F'}
\end{lemma}

\begin{proof}
  We will show that (\ref{ineq-sum}) applies to $F'$. Fix a clause $C
  \in F'$. After the sparsification process, every literal $u$
  fulfills $\sum_{D\in \mathcal{G'} : u \in D }\Pr[\alpha \not \models
  D] \leq \frac{1}{8k}$.  Therefore, the terms $\Pr[\alpha \not
  \models D]$, for all good clauses $D$ conflicting with $C$, sum up to at
  most $\frac{1}{8}$. By Lemma~\ref{prop-bad}, the terms $\Pr[\alpha
  \models D]$ for all bad clauses $D$ also sum up to at most
  $\frac{1}{8}$.  Hence (\ref{ineq-sum}) holds, and by
  Lemma~\ref{corollary-llll}, $F'$ is satisfiable, and clearly
  $F$ as well.  
\end{proof}

Contrary, if $F$ is unsatisfiable, the sparsification process produces
the empty clause.  We will show that $e(F)$ is large. There is some $C
\in \mathcal{G}$ all whose literals are being deleted during the
sparsification process. Write $C = \{u_1, u_2, \dots, u_k\}$, and
order the $u_i$ such that $\occ_F(u_1) \leq \occ_F(u_2)\leq \dots \leq
\occ_F(u_k)$. One checks that this implies that $p(u_1) \leq p(u_2)
\leq \dots \leq p(u_k)$.  Fix any $\ell \in \{1,\dots,k\}$ and let
$u_j$ be the first literal among $u_1,\dots,u_\ell$ that is deleted
from $C$. Let $C'$ denote what is left of $C$ just before that
deletion, and consider the set $\mathcal{G}'$ at this point of time.
Then $\{u_1,\dots,u_\ell\} \subseteq C' \in \mathcal{G}'$. By the
definition of the process,
  \begin{eqnarray*}
    \frac{1}{8k} & < & \sum_{D \in \mathcal{G}' :\ u_j \in D} \Pr[\alpha 
    \not \models D] \leq
    \sum_{D\in \mathcal{G}' :\ u_j \in D } \Pr[\alpha 
    \not \models C'] \leq \\
    & \leq & \occ_F(u_j) \Pr[\alpha \not \models C'] \leq \\
    & \leq & \occ_F(u_\ell) \prod_{i=1}^{\ell} (1-p(u_i)) \ .
  \end{eqnarray*}
  
  Since $p(u) \geq \sqrt[k]{\frac{\occ_F(u)}{16e(F)}}$ for all
  literals $u$ in a good clause, it follows that $\frac{1}{128ke(F)}
  \leq p(u_\ell)^k\prod_{i=1}^{\ell} (1-p(u_i))$,
  for every $1\leq \ell \leq k$.\\
  
  Let $(q_1,\dots,q_k) \in [\frac{1}{2},1]^k$ be any sequence
  satisfying the $k$ inequalities $\frac{1}{128ke(F)} \leq q_\ell^k
  \prod_{i=1}^{\ell} (1-q_i)$ for all $1\leq\ell\leq k$, for example,
  the $p(u_i)$ are such a sequence. We want to make the $q_\ell$ as
  small as possible: If (i) $q_\ell > \frac{1}{2}$ and (ii)
  $\frac{1}{128ke(F)} < q_\ell^k \prod_{i=1}^{\ell} (1-q_i)$, we can
  decrease $q_\ell$ until one of (i) and (ii) becomes an equality.
  The other $k-1$ inequalities stay satisfied. In the end we get a
  sequence $q_1, \dots, q_k$ satisfying $\frac{1}{128ke(F)} = q_\ell^k
  \prod_{i=1}^{\ell} (1-q_i)$ whenever $q_\ell > \frac{1}{2}$. This
  sequence is non-decreasing: If $q_\ell > q_{\ell+1}$, then
  $q_\ell>\frac{1}{2}$, and $\frac{1}{128ke(F)} \leq
  q_{\ell+1}^k\prod_{i=1}^{\ell+1}(1-q_i) <
  q_{\ell}^k\prod_{i=1}^{\ell}(1-q_i) = \frac{1}{128ke(F)}$, a
  contradiction. \\

  If all $q_i$ are $\frac{1}{2}$, then the $k$\textsuperscript{th}
  inequality yields $128ke(F) \geq 4^k$, and we are done. Otherwise,
  there is some $\ell^* = \min\{i \ | \ q_i > \frac{1}{2}\}$. For
  $\ell^* \leq j < k$ both $q_j$ and $q_{j+1}$ are greater than
  $\frac{1}{2}$, thus
  $q_{j+1}^k\prod_{i=1}^{j+1}(1-q_i)=q_{j}^k\prod_{i=1}^{j}(1-q_i)$,
  and $q_j = q_{j+1} \sqrt[k]{1-q_{j+1}}$. We define
  $$
  f_k(t) := t \sqrt[k]{1-t} \  ,
  $$
  thus $q_j = f_k(q_{j+1})$. By $f_k^{(j)}(t )$ we denote $f_k (f_k (
  \dots (f_k(t)) \dots))$, the $j$-fold iterated application of
  $f_k(t)$, with $f_k^{(0)}(t)=t$. We obtain $q_j = f_k^{(k-j)}(q_k) >
  \frac{1}{2}$ for $\ell^* \leq j \leq k$.  By Part (v) of
  Proposition~\ref{prop-ell}, $f_k^{(k-1)}(q_k) \leq \frac{1}{2}$,
  thus $\ell^* \geq 2$. Therefore $q_1 = \dots= q_{\ell^* - 1} =
  \frac{1}{2}$, and the $(l^*-1)$\textsuperscript{st} inequality reads
  as
  $$\frac{1}{128ke(F)} \leq q_{\ell^*-1}^k \prod_{i=1}^{\ell^*-1} (1-q_i)
  = 2^{-k-\ell^*+1} \ .
  $$
  We obtain $e(F) \geq \frac{2^{k+\ell^*-1}}{128k}$. How large is
  $\ell^*$? Define $S_k := \min\{\ell\in\mathbb{N}_0 \ | \
  f_k^{(\ell)}(t) \leq \frac{1}{2}\ \forall t \in [0,1]\}.$ By Part
  (v) of Proposition~\ref{prop-ell} (see appendix), $S_k$ is finite.
  Since $f_k^{(k-\ell^*)}(q_1) = q_{\ell^*} > \frac{1}{2}$, we
  conclude that $k - \ell^* \leq S_k-1$, thus $e(F) \geq \frac{2^{2k -
      S_k}}{128k}$.
  
  \begin{lemma}
    The sequence $\frac{S_k}{k}$ converges to 
    $
    \lim_{k \rightarrow \infty} \frac{S_k}{k} 
    = - \int_{\frac{1}{2}}^1 \frac{1}{x \ln (1-x)} dx < 0.572 
    $.
    \label{prop-alpha}
  \end{lemma}
  
  The proof of this lemma is technical and not related to
  satisfiability. We prove it in the appendix. We conclude that $e(F)
  \geq \frac{2^{(2-0.572)k}}{128k} \in \Omega\left(2.69^k\right)$.
\end{proof}

\section{Conclusion}

We want to give some hindsight why a sparsification procedure is
necessary in both lower bound proofs in this paper.  The probability
distribution we define is not a uniform one, but biased towards
setting $x$ to \texttt{true} if $\occ_F(x) \gg \occ_F(\bar{x})$.  Let
$C$ be a clause containing $\bar{x}$. It conflicts with all clauses
containing $\bar{x}$. It could happen that in all those clauses, $x$
is the only literal with $p(x) > \frac{1}{2}$. In this case, each such
clause is unsatisfied with probability not much smaller than $2^{-k}$,
and the sum (\ref{ineq-sum}) is greater than $\frac{1}{4}$. By
removing $x$ from these clauses, we reduce the number of clauses
conflicting with $C$, making the sum (\ref{ineq-sum}) much smaller.
However, for other clauses $C'$, this sum might increase by removing
$x$. We think that one will not be able to prove a tight lower bound
using just a smarter sparsification process. We state some open
problems and questions.

\begin{quotation}
{\em Question:} Does $\lim_{k \rightarrow \infty} \sqrt[k]{gc(k)}$ exist?
\end{quotation}

If it does, it lies between $2.69$ and $3.51$. One way to prove
existence would be to define ``product'' taking a $k$-CNF formula $F$
and an $\ell$-CNF formula $G$ to a $(k+\ell)$-CNF formula $F \circ G$
that is unsatisfiable if $F$ and $G$ are, and $e(F \circ G) =
e(F)e(G)$. With $2$ and $4$ ruled out, there seems to be no obvious
guess for the value of the limit. What about $\sqrt{8}\approx 2.828$,
the geometric mean of $2$ and $4$? 

\begin{quotation}
  {\em Question:} Is there an $a > 2$ such that every unsatisfiable
  $k$-CNF formula contains a variable $x$ with $\occ_F(x) \cdot
  \occ_F(\bar{x}) \geq a^k$?
\end{quotation}

Where do our methods fail to prove this? The part in the proof of the
lower bound of Theorem~\ref{main} that fails is Lemma~\ref{prop-bad}.
On the other hand, Lemma~\ref{prop-bad} proves more than we need for
Theorem~\ref{main}: It proves that $\Pr[\alpha \models D]$, summed up
over {\em all} bad clauses gives at most $\frac{1}{8}$. We only need
that the bad clauses conflicting with a specific clause sum up to at
most $\frac{1}{8}$. Still, we do not see how to apply or extend our
methods to prove that such an $a > 2$ exists.\\

We discussed lower and upper bounds on the minimum of several parameters
of unsatisfiable $k$-CNF formulas. The following table lists them
where bounds labeled with an asterisk are from this paper and
unlabeled bounds are not attributed to any specific paper. 

\begin{table}[htbp]
  \centering
  \begin{tabular}{lllll}
    parameter & notation & lower bound & upper bound\\\hline
    occurrences of a literal & $\occ(x)$ & 1 & 1\\
    occurrences of a variable & $f(k)$ & 
    $\frac{2^k}{ek}$~\cite{KST1993} & $\frac{2^{k+3}}{k}$~\cite{Gebauer2009}\\
    local conflict number & $lc(k)$ & $\frac{2^k}{e}$~\cite{KST1993} & $2^k-1$\\
    conflicts caused by a variable \hspace{2mm} & $\occ(x)\occ(\bar{x})$ & 
    $\frac{2^k}{ek}$~\cite{KST1993} & $O(3.01^k)^*$\\
    global conflict number & $gc(k)$ & $\Omega(2.69^k)^*$ & $O(3.51^k)^*$
  \end{tabular}
\end{table}

\bibliographystyle{abbrv}
\bibliography{refs}

\newpage
\appendix

\section{Proof of Lemma~\ref{prop-alpha}}

\begin{proposition}
  Let $k \in \mathbb{N}$ and $f_k : [0,1] \rightarrow [0,1]$ with $f_k(t)=t\sqrt[k]{1-t}$. For $t \in [0,1]$, the following statements hold.
  \begin{itemize}
    
  \item[(i)] $f_k(t)$ attains its unique maximum at
    $t = t^*_k := \frac{k}{k+1}$.
  \item[(ii)] $f_k(t) \leq t$, and $f_k(t)=t$ if and only if $t=0$.
  \item[(iii)] For $\ell \geq 1$, $f^{(\ell)}_k(t) \leq
    f^{(\ell)}_k\left(\frac{k}{k+1}\right)$.
  \item[(iv)] For $\ell \geq 0$ and $t \in [0,1]$,
    $(1-t)^{\ell / k}\ t \leq f_k^{(\ell)}(t) \leq
    (1-f_k^{(\ell)}(t))^{\ell/k}\ t$.
  \item[(v)] For $k\geq 2$ and any $t \in [0,1]$, $f_k^{(k-1)}(t)
    \leq \frac{1}{2}$.
  \end{itemize}
  \label{prop-ell}
\end{proposition}

\begin{proof}
  $(i)$ follows from elementary calculus. $(ii)$ holds since
  $\sqrt[k]{1-t}$ is less than $1$ for all $t > 0$. For $\ell=1$,
  $(iii)$ follows from $(i)$, and for greater $\ell$, it follows from
  $(ii)$ and induction on $\ell$. $(iv)$ holds because each of the
  $\ell$ applications of $f_k$ multiplies its argument with a factor
  that is at least $\sqrt[k]{1-t}$ and at most $\sqrt[k]{1 -
    f^{(\ell)}_t}$.  Suppose $(v)$ does not hold. Then by (iii) we get
  $f^{(k-1)}_k\left(\frac{k}{k+1}\right)\geq f_k^{(k-1)}(t) >
  \frac{1}{2}$, and by $(iv)$, we have
  $$
  \frac{1}{2} < f_k^{(k-1)}\left(\frac{k}{k+1}\right)
  \leq \left(\frac{1}{2}\right)^{\frac{k-1}{k}} \frac{k}{k+1} \ .
  $$
  An elementary calculation shows that this does not hold for any 
  $k \geq 1$.
\end{proof}

% \begin{figure}
%   \begin{center}
%   \includegraphics[width=0.5\textwidth]{function.eps}
%   \caption{The graph of $f_k(x)$ for $k=4$. The function has a unique
%     maximum, and $f_k(x) \leq x$ for all $x \in [0,1]$.}
%   \end{center}
% \end{figure}

To prove Lemma~\ref{prop-alpha}, we compute $\lim_{k \rightarrow
  \infty} \frac{S_k}{k}$ (and show that the limit exists). Recall
the definition
$$
S_k = \min\{\ell \in \mathbb{N}_0 \ | \ f_k^{(\ell)}(t) \leq
\frac{1}{2} \ \forall t \in [0,1] \} \ ,
$$
where $f_k(t)=t\sqrt[k]{1-t}$.  By Part $(iii)$ of
Proposition~\ref{prop-ell}, $S_k =\min\{\ell \ | \
f_k^{(\ell)}(t^*_k) \leq \frac{1}{2}\}$, for $t^*_k := \frac{k}{k+1}$. We
generalize the definition of $S_k$ by defining for $t \in (0,1]$,
$$
S_k(t) := \min \{\ell \ | \ f_k^{(\ell)} (t^*_k) \leq t\} \ .
$$
Further, we set $s_k(t) := \frac{S_k(t)}{k}$.
 Let $0 < t_2 < t_1 < t^*_k$. We want to estimate
$s_k(t_2) - s_k(t_1)$. This should be small if $|t_1-t_2|$ is small.
For brevity, we write $a := S_k(t_1)$, $b := S_k(t_2)$. Clearly $a
\leq b$. We calculate
\begin{eqnarray*}
  t_2 & \geq & f_k^{(b)}(t^*_k) = f_k^{(b-a+1)} (f_k^{(a-1)}(t^*_k))
  \geq f_k^{(b-a+1)}(t_1) \geq (1-t_1)^{((b-a+1)/k)} t_1 \ ,\\
  t_2 &  <& f_k^{(b-1)}(t^*_k) = f_k^{(b-a-1)}(f_k^{(a)}(t^*_k))
  \leq f_k^{(b-a-1)}(t_1) \leq (1-t_2)^{((b-a-1)/k)} t_1 \ .
\end{eqnarray*}
Where we used part $(iv)$ of Proposition~\ref{prop-ell}.  In fact, these
inequalities also hold if $t_1 \geq t^*_k$, when $a=0$:
\begin{eqnarray*}
  t_2 & \geq & f_k^{(b)}(t^*_k) \geq (1-t^*_k)^{b/k}t_1 
  \geq (1-t_1)^{(b+1)/k}t_1\ , \\
  t_2 & < & f_k^{(b-1)}(t^*_k) = (1-t_2)^{((b-1)/k)} t_1 \ .
\end{eqnarray*}
One checks that the inequalities even hold if $t^* \leq t_2 < t_1 \leq
1$. Note that $\frac{b-a}{k}=s_k(t_2)-s_k(t_1)$.  Solving for
$\frac{b-a}{k}$, the above inequalities yield
\begin{eqnarray}
\frac{\log t_2 - \log t_1}{\log (1-t_1)} - \frac{1}{k} \leq s_k(t_2) -
s_k(t_1) \leq \frac{\log t_2 - \log t_1}{\log (1-t_2)}+ \frac{1}{k} \ ,
\label{ineq-sk}
\end{eqnarray}
for all $0 < t_2 < t_1 < 1$. The right inequality also holds for $0 <
t_2 < t_1 \leq 1$. Multiplying with $-1$, we see that it also holds if
$t_2 > t_1$. If $t_2 = t_1$, it is trivially true.
Hence this inequality is true for all $t_1,t_2 \in (0,1)$.\\

Suppose $s(t) =
\lim_{k\rightarrow\infty} s_k(t)$ exists, for every fixed $t$.
Inequality (\ref{ineq-sk}) also holds in the limit. Writing $t_1=t$
and $t_2 = t+h$ and dividing (\ref{ineq-sk}) by $h$ gives
\begin{eqnarray*}
  \frac{\log (t+h) - \log t}{h\log (1-t)} \leq 
  \frac{s(t+h) - s(t)}{h}
  \leq \frac{\log (t+h) - \log t}{h\log (1-t-h)}\ ,
\end{eqnarray*}
Letting $h$ go to $0$, we obtain $s'(t) = \frac{1}{t\log(1-t)}$, thus
$s(t) = s(1) - \int_t^1 \frac{1}{x\log(1-x)} dx$. Observing that
$\frac{S_k}{k} = s_k(\frac{1}{2})$ and $s_k(1)=0$ for all $k$
proves the Lemma. \\

The above argument shows that if $s_k(t)$ converges pointwise, then it
converges to a continuous function $s(t)$ on $(0,1)$. We have to show
that $\lim_{k\rightarrow \infty} s_k(t)$ does in fact exist.  First
plug in $t_1 = 1$ into the right inequality of (\ref{ineq-sk}) to
observe that for each fixed $t_2$, the sequence $(s_k(t_2))_{k \in
  \mathbb{N}}$ is bounded from above. Clearly it is bounded from below
by $0$. Hence there exist $\us(t) := \limsup s_k(t)$ and similarly
$\ls(t):= \liminf s_k(t)$. We write shorthand $L(t_1,t_2):=\frac{\log
  t_2 - \log t_1}{\log (1-t_1)}$ and $U(t_1,t_2):=\frac{\log t_2 -
  \log t_1}{\log (1-t_2)}$. Now (\ref{ineq-sk}) reads as 
$L(t_1,t_2)-\frac{1}{k}  \leq  s_k(t_2) - s_k(t_1) \leq
U(t_1,t_2)+\frac{1}{k}$. We claim that
\begin{eqnarray}
L(t_1,t_2) & \leq & \us(t_2) - \us(t_1) \leq U(t_1,t_2) \ , \label{ineq-us}\\
L(t_1,t_2) & \leq & \ls(t_2) - \ls(t_1) \leq U(t_1,t_2) \ . \label{ineq-ls}
\end{eqnarray}

For sequences $(a_k)_{k \in \mathbb{N}}$, $(b_k)_{k \in \mathbb{N}}$,
$\limsup a_k - \limsup b_k = \limsup (a_k - b_k)$ does not hold in
general, hence the claim is now completely trivial. We will proof that
$\us(t_2) -  \us(t_1) \leq U(t_1,t_2)$.  This will
prove one claimed inequality. The other three inequalities can be
proven similarly.  Fix some small $\epsilon > 0$.  For all
sufficiently large $k$, $\frac{1}{k} \leq \epsilon$.  We have
$s_k(t_2) \geq \us(t_2) - \epsilon$ for infinitely many $k$, thus
$s_k(t_1) \geq s_k(t_2) - U(t_1,t_2) - \frac{1}{k} \geq \us(t_2) -
U(t_1,t_2) - 2\epsilon$ for infinitely many $k$.  Therefore $\us(t_1)
\geq s(t_2) - U(t_1,t_2) - 2\epsilon$.  By making $\epsilon$
arbitrarily small, the claimed inequality follows.\\

We can now apply our non-rigorous argument from above, this time
rigorously. Write $t=t_1$, $t_2 = t+h$, and divide (\ref{ineq-us}) and
(\ref{ineq-ls}) by $h$, send $h$ to $0$, and we obtain $\us'(t) =
\ls'(t) = \frac{1}{t\log(1-t)}$. Since $\us(1)=\ls(1)=0$, we obtain
$$
\us(t) = \ls(t) = \int_t^1 \frac{-1}{x\log(1-x)}dx \ .
$$

\end{document}